\documentclass[preliminary,copyright,creativecommons]{eptcs}

\usepackage{aiml26}

\usepackage{iftex}

\ifpdf
  \usepackage[T1]{fontenc}        
\else
  \usepackage{breakurl}           
\fi

\usepackage[utf8]{inputenc}
\usepackage[T1]{fontenc}
\usepackage{amsfonts}
\usepackage{graphicx}
\usepackage{amsmath}
\usepackage{amssymb}
\usepackage[dvipsnames]{xcolor}
\usepackage{float}
\usepackage{esvect}
\usepackage{mathtools}
\usepackage{subcaption}
\usepackage{stmaryrd}
\usepackage{comment}
\usepackage{wasysym}
\usepackage{stmaryrd}

\DeclareMathAlphabet{\mathcal}{OMS}{cmsy}{m}{n}
\usepackage{tikz}
\usetikzlibrary{positioning,arrows,shapes,snakes,automata,backgrounds,petri}
\usepackage{forest}
\usepackage{caption}
\catcode`@=11
\def\diamondplus{\mathbin{\mathpalette\diamondplus@\relax}}
\def\diamondplus@#1#2{%
  \vcenter{%
    \hbox{%
      \setbox\z@=\hbox{$\m@th#1\oplus$}%
      \dimen@=\ht\z@ \advance\dimen@ \dp\z@
      \resizebox{!}{\dimen@}{%
        \rotatebox[origin=c]{45}{$\m@th#1\boxtimes$}%
      }
    }
  }
}
\catcode`@=12

\newcommand{\nigi}[1]{\textcolor{magenta}{Nina says: #1}}

\title{Strategies in Sabotage Games:\\ Temporal and Epistemic Perspectives}
\author{Nina Gierasimczuk \qquad\qquad Katrine B. P. Thoft 
\institute{Technical University of Denmark\\
Kgs. Lyngby, Denmark}
\email{\quad nigi@dtu.dk \quad\qquad kabpt@dtu.dk}
}

\newcommand{\titlerunning}{Strategies in Sabotage Games}
\newcommand{\authorrunning}{N. Gierasimczuk \& K.B.P. Thoft}

\hypersetup{
  bookmarksnumbered,
  pdftitle    = {\titlerunning},
  pdfauthor   = {\authorrunning},
  pdfsubject  = {EPTCS},               
  pdfkeywords = {keyword1, keyword2} 
}

\begin{document}
\maketitle

\begin{abstract}
Sabotage games are played on a dynamic graph, in which one agent, called a runner, attempts to reach a goal state, while being obstructed by a demon who at each round removes an edge from the graph. Sabotage modal logic was proposed to carry out reasoning about such games. Since its conception, it has undergone a thorough analysis (in terms of complexity, completeness, and various extensions) and has been applied to a variety of domains, e.g., to formal learning. In this paper, we propose examining the game from a temporal perspective using alternating time temporal logic (ATL$^\ast$), and address the players' uncertainty in its epistemic extensions. This framework supports reasoning about winning strategies for those games, and opens ways to address temporal properties of dynamic graphs in general.
\end{abstract}

\section{Introduction}
Many real-world problems can be modelled as graphs undergoing structural changes. A common example is the connectivity of a train transportation network: \textit{Can I reach Copenhagen from Amsterdam despite a train connection being cancelled?} Being able to answer such questions is crucial for making sense of agents' behaviour in a network. Sabotage games are an abstract model of such scenarios: one player (the runner) attempts to reach a specified target vertex in a graph, while an adversary (the blocker) removes edges in the graph. Traditionally, sabotage games have been analyzed using Sabotage Modal Logic (SML) introduced in \cite{Benthem:2005,Loding:2003}. Returning to our train network example, in SML the possibility of travelling to Copenhagen ($Cph$) from Amsterdam, after a cancellation of a train connection, is expressed as $\blacklozenge\lozenge Cph$ being true while in Amsterdam, which says: after a deleting an arbitrary edge in the train network, there will (still) be a connection between Amsterdam and Copenhagen. 

The elegant approach of SML does not account of temporal and epistemic aspects of sabotage games (see \cite{vanBenthem:2010}).
The overall aim of this paper is to provide a unified logical framework where the strategic temporal reasoning native to Alternating-time Temporal Logic (ATL, \cite{Alur:2002}), and the adversarial structural change in SML interact, in order to study how sabotage unfolds over time in dynamic games with imperfect information. 

The paper proceeds as follows. In Section \ref{sec:sabotage_games} we present Sabotage Games: the classical (reachability) version and the newly introduced liveness sabotage game. We also present SML characterizations of the existence of winning strategies. In Section \ref{sec:temporal} we reconstruct Sabotage Games in Alternating-time Temporal Logic setting, as sabotage game structures, and we show how to express the existence of winning strategies ATL$^\ast$. We also discuss how the strategic temporal setting of ATL$^\ast$ allows expressing the graph-theoretical notion of (dynamic) minimal $s-t$ cut of a graph. We also discuss an extension to angelic sabotage games that feature the addition of edges. In Section \ref{sec:atel}, we show how the epistemic extension of ATL can be used to reason about knowledge of strategic ability in sabotage games. Finally, in Section \ref{sec:conclusion}, we conclude and outline several directions of future work.

\section{Sabotage games: reachability and liveness}\label{sec:sabotage_games}
The classic sabotage game was first introduced in \cite{Benthem:2005} and \cite{Loding:2003} as a turn-based two-player game $(G,v_0, v_g)$ played on a graph $G=(V,E)$, with $v_0,v_g\in V$ being some pre-specified start and goal vertices. 
The first player in the game is the runner, who is positioned at the designated start vertex and continues moving along the edges of the graph until she reaches either the goal vertex $v_g$ or she arrives at a dead-end. The runner can move along exactly one edge in each turn. The second player is the demon (often called `blocker'), who acts as a saboteur and, in each turn, removes an arbitrary edge from the graph. Since the aim of the runner is to reach the goal, in this paper we will call such games \emph{reachability sabotage games} (RSG, for short). Motivated by concurrency theory \cite{Murata:1989,Chatterjee:2007}, apart from the the RSGs, in this paper we introduce and study another variant of the game: \emph{liveness sabotage games} (LSG, for short). Instead of the pre-specified goal vertex, LSG has a \emph{liveness parameter} $b\in \mathbb{N}$, that denotes the number of moves the runner should (at least) be able to make before she gets completely blocked by the demon, i.e., a lower bound of how long the runner can stay `alive'. An LSG is then specified by $(G,v_0,b)$, where $G=(V,E)$, $v_0\in V$ is the start vertex, and $b$ is the liveness parameter. 
Throughout this paper we will assume the underlying $G=(V,E)$ to be a directed graph with $V$ and $E$ non-empty, we designate some $v_0\in V$ to be the starting vertex and sometimes $v_g\in V$ to be the goal vertex. We will use $|\cdot|$ to denote the cardinality of a set. We will also assume a fixed enumeration of vertices from $V\setminus\{v_0\}$, $(v_i)_{i\in\{1,\ldots, |V|-1\}}$, and a fixed enumeration of edges from $E$, $(e_j)_{j\in\{1,\ldots, |E|\}}$. 

The sabotage game proceeds through a sequence of game-states---each can be characterized by a set of (remaining) edges and a distinguished vertex corresponding to runner's current position.
\begin{definition}\label{def:states}
The set of \emph{sabotage game-states} of $G=(V,E)$ is $\mathcal{S}(G)=\{(E',v)\mid E'\subseteq E, v\in V\}$.\footnote{Whenever $G$ is clear from context, we will drop the argument $G$ and refer to the set of game-states with $\mathcal{S}$.}
 
\end{definition} 

A sabotage play is a specific sequence of such game-states.

\begin{definition}\label{def:play}
Let $G=(V,E)$ and $v_0\in V$ be a starting vertex. A \emph{strictly turn-based sabotage $v_0$-play} is a sequence of game-states $s^0,\ldots, s^n$, such that $s^0=(E^0,v^0)$ is given by $(E,v_0)$, and for any $k$, such that $0<k\leq n$: $$s^{k}=
   \begin{cases}
   (E^{k-1},v^{k}),\text{  s.t.  } (v^{k-1},v^{k})\in E^{k-1} & \text{if } k \text{ odd},\\
   (E^{k-1}\setminus \{(x,y)\},v^{k-1}),\text{ s.t. }(x,y)\in E^{k-1} & \text{otherwise}.\\
   \end{cases}
   $$ 
   \end{definition}
A sabotage match is a sabotage play whose final game-state is the first one in the play that satisfies the winning conditions of the game, different for reachability and liveness versions of the game.

\begin{definition}\label{def:match}
\item A \emph{turn-based sabotage match} is a turn-based sabotage play $s^0,\ldots,s^n$, with $s^n=(E^n,v^n)$ such that: 
   \begin{enumerate}
\item \emph{for $RSG$ $(G,v_0,v_g)$:} for all $k<n$ we have $v^k\neq v_g$, and either $v^n=v_g$ (i.e., runner wins), or $v^n\neq v_g$ and there is no $v\in V$ such that $(v^n,v)\in E^n$ (i.e., demon wins);
\item \emph{for $LSG$ $(G,v_0,b)$:} either $n\geq b$ (runner wins), or $n<b$ and there is no $v\in V$ such that $(v^n,v)\in E^n$ (i.e., demon wins).
\end{enumerate}
\end{definition}

\paragraph{Sabotage Modal Logic}\label{sec:sab_log}

After establishing some preliminary intuitions about reachability and liveness turn-based sabotage games, we are now ready to talk about the logic customarily used to describe them. Let us briefly recall the classical Sabotage Modal Logic (SML, \cite{Benthem:2005,Loding:2003}).

\textbf{SML Syntax} Let $\Phi$ be a finite set of propositions and $p\in \Phi$. The syntax of SML is given by:
$\varphi::= \top \, |\, p \,| \,\neg \varphi \,| \, \varphi \land \varphi\,
| \, \lozenge \varphi\,
| \, \blacklozenge \varphi $. Standardly, we will also use $\bot$ for $\neg \top$, $\varphi\vee\psi$ for $\neg(\neg \varphi \wedge \neg \psi)$, $\Box \varphi$ for $\neg\lozenge\neg\varphi$, and $\blacksquare \varphi$ for $\neg \blacklozenge \neg \varphi$. As usual, the formula $\lozenge \varphi$ stands for the possibility of transitioning to an accessible vertex where $\varphi$ holds. The $\blacklozenge$-operator, the sabotage modality, denotes the removal of an \emph{arbitrary} edge in the graph.

\textbf{SML Semantics}
The language of SML is interpreted over \emph{sabotage models} $M=(W,R,\textup{\texttt{Val}})$, where $W$ is a finite non-empty set of worlds, $R\subseteq W\times W$ is a binary relation over $W$, and $\textup{\texttt{Val}}: \Phi \rightarrow \mathcal{P}(W)$. The truth of the formulas of SML is defined locally at pairs $(M,s)$. Let $M=(W,R,\textup{\texttt{Val}})$ be a sabotage model, $v\in W$ and $p\in \Phi$. The clauses for the standard ML part of the language are as usual, and for $\blacklozenge$ we have:
$ M,v \models \blacklozenge \varphi  \text{ iff there is a  } (x,y)\in R \text{ s.t. } M^{-}_{(x,y)},v \models \varphi, $
where $M^{-}_{(x,y)}$ is the updated model after deleting an edge between $x$ and $y$, i.e.,  $M^{-}_{(x,y)}=(W,R\setminus\{(x,y)\},V)$. In other words, the formula $ \blacklozenge \varphi$ is true in a model $M$ and a state $v$ if, after removing some edge, $\varphi$ is true at the resulting model at the state $v$.

SML can be used to express properties of sabotage games. Let the set of propositions be $\Phi=\{r,g\}$. Any game-state $s=(E,v)$ of a sabotage game can be transformed into a sabotage model $M(s)=(V, E, \texttt{Val})$, where $\texttt{Val}(r)=\{v\}$, and (in the case of $RSG$ where a goal vertex is specified) $\texttt{Val}(g)=\{v_g\}$. Intuitively speaking, $g$ labels the goal vertex, and $r$ stands for the current position of the runner.\footnote{Since these propositions are true in unique worlds, they can be seen as index propositions of \emph{hybrid logic} (see, e.g., \cite{blackburn1995hybrid}). Note however that multiple worlds making $g$ and $r$ true could, in a natural way, represent distributed goals and multiple runners.}  
In \cite{Benthem:2005}, it was observed that, given this represention, SML allows expressing the existence of a winning strategy for the runner starting at the initial vertex $v_0$. 
    

\begin{proposition}[\cite{Benthem:2005}]\label{prop:reachabilityWin}
Let $G=(V,E)$ with $|E|=k$. Runner has a winning strategy in $RSG$ $((V,E),v_0,v_g)$ iff $M(s^0),v_0\models \rho_k$, where $\rho_i$ (for $i\in \mathbb{N}$) is defined inductively as $\rho_0:=g, \ \rho_{n+1}:= g\vee \lozenge\blacksquare\rho_n$.\footnote{The proof of this proposition, together with different interpretations and variations of this game can be found in \cite{Gierasimczuk:2009}.}
\end{proposition} 
A similar SML characterization be obtained for our newly introduced liveness sabotage game.

\begin{proposition}\label{prop:livenessWin}
Let $G=(V,E)$ and $b\in \mathbb{N}^+$. Runner has a winning strategy in $LSG=((V,E),v_0,b)$ iff $M((E,v_0)),v_0\models \gamma_b$, with $\gamma_i$ (for $i\in \mathbb{N}$) given by:
  $ \gamma_1:=\lozenge\top, \ \gamma_{n+1}:=\lozenge \blacksquare\gamma_n.$\footnote{The proofs of propositions can be found in Appendix.}

\end{proposition}

The alternating structure of the modal prefix in the formulae used in Prop.~\ref{prop:reachabilityWin} and \ref{prop:livenessWin} is specific to strictly turn-based sabotage games. While expressing the winning conditions for different sequences of moves (e.g., allowing several moves of runner before demon moves), can be easily expressed in SML by iterating modalities of one kind, but each such protocol requires a different SML formula. To gain a more general perspective, it would be useful to be able to refer the strategies directly in the logical language.  
 The inductively defined formulae $\rho$ and $\gamma$ indicate that interesting properties of sabotage games hide in the time-progression of the game. After all, in theoretical computer science, `reachability' (as in RSG) and `liveness' (as in LSG) are considered \emph{temporal} properties. All this motivates the use of temporal logic to study sabotage games. Our focus on Alternating-time Temporal Logic is additionally justified by a limitation of the traditional approach---simultaneous moves of the players are beyond the scope of SML.

\vspace{-0.3cm}
\section{Temporal perspective on sabotage games}\label{sec:temporal}
In ATL and ATL$^\ast$, branching-time logic is enriched with strategic operators that allow quantifying over paths (or plays) resulting from players' strategies. ATL$^\ast$ is a multi-agent extension of CTL$^\ast$ (just as ATL is a multi-agent extension of CTL). We will now reconstruct sabotage games in the setting of alternating time temporal logic ATL$^\ast$, as sabotage game structures. Intuitively, the main difference with sabotage models is that the states of the sabotage game structures represent the whole game-state, rather than a concrete vertex in the graph. 

\paragraph{Sabotage game structures}
We will distinguish two basic kinds of sabotage games: strictly turn-based (just as the classical sabotage game discussed in Section~\ref{sec:sabotage_games}) and concurrent (in which the agents move simultaneously). Two componenents will be fixed and the same across the different types: the set of agents containing runner and demon, $\mathbb{AG}=\{r,d\}$, and the set of actions. The crucial observation about sabotage games is that both players act on the \emph{edges of the graph}. Of course, their actions have different effects: the demon choosing an edge results in a deletion of that edge, while the runner choosing an edge results in updating the runners position in the graph. We will then take the set of action to be $Act=\{e\mid e\in E\}$ (when we allow an agent to be inactive, we will extend the set with a special skip-symbol: $Act^{skip}=Act\cup\{\textup{\texttt{skip}}\}$). We will address the qualitative difference between the outcomes of the two players edge-choices when defining the transition function. To be able to `extract' vertices from our edge-notation we will at times use the projection function, so that if $e=(x,y)$, then $\pi_1(e)=x$ and $\pi_2(e)=y$.  
\begin{definition}\label{def:TurnSabotageGameStructure} 
    A \emph{turn-based \textup{(\texttt{tb-})}sabotage game structure} based on $G=(V,E)$ is a tuple: $\mathbb{S}^{\textup{\texttt{tb}}}(G) =(\mathbb{AG},Act^{\textup{\texttt{skip}}},S^{\textup{\texttt{tb}}},act^{\textup{\texttt{tb}}},\delta^{\textup{\texttt{tb}}}),$ where: $S^{\textup{\texttt{tb}}}=\mathcal{S}(G)\times \mathbb{AG}$ is a finite set of (duplicated) game-states, each labeled with either runner or demon; $act^{\textup{\texttt{tb}}}:\mathbb{AG}\times S^{\textup{\texttt{tb}}} \rightarrow \mathcal{P}(Act)$ is the function that determines which actions are allowed to be performed by each agent in a game-state. Specifically, given a state $t=((E',v),a)$: 
        \noindent\begin{itemize}
\item if $a=r$, then $act^{\textup{\texttt{tb}}}(r,t)=\{e\in E'\mid \pi_1(e)=v\}$ and $act^{\textup{\texttt{tb}}}(d,t)=\{\textup{\texttt{skip}}\}$;
        \item if $a=d$, then $act^{\textup{\texttt{tb}}}(d,t)=E'$ and $act^{\textup{\texttt{tb}}}(r,t)=\{\textup{\texttt{skip}}\}$; 
        \end{itemize}
Finally, $\delta^{\textup{\texttt{tb}}}$ is a transition function that assigns to any $t=((E',v),a)$ in $S^{\textup{\texttt{tb}}}$ and action profile $(e_r,e_d)$, such that $e_i\in act(i,t)$ for $i\in\{r,d\}$, a unique successor state: 
        $$\delta^{\textup{\texttt{tb}}}(t,(e_r,e_d))= \begin{cases} 
        ((E'\setminus{e_d},v),r)  & \text{if } e_r=\textup{\texttt{skip}} ;  \\ 
        ((E',\pi_2(e_r)),d) & \text{otherwise}. \end{cases}$$
\end{definition}
\noindent The above definition enforces the `turn-based' form by marking states with their `owners'. Then, the `forced' \texttt{skip} is used by runner in a game-state owned by demon (and \emph{vice versa}). 

The ATL-framework is naturally well-suited to analyze concurrent game structures, which formalize situations were players (attempt to) move simultaneously. We will assume that the conflicting choice of selecting the same edge results in mutual cancellation and has no effect.

\begin{definition}\label{def:concurrentSabotageGameStructure} 
    A \emph{concurrent \textup{(\texttt{con}-)}sabotage game structure} based on a graph $G=(V,E)$ is a tuple: $\mathbb{S}^{\textup{\texttt{con}}}(G) =(\mathbb{AG},Act,S^{\textup{\texttt{con}}},act^{\textup{\texttt{con}}},\delta^{\textup{\texttt{con}}}),$ where:
$S^{\textup{\texttt{con}}}= \mathcal{S}(G)$;
$act^{\textup{\texttt{con}}}:\mathbb{AG}\times S^{\textup{\texttt{con}}} \rightarrow \mathcal{P}(Act)$, s.t. given an $s=(E',v)$: $act(d,s)=E'$, and 
$act(r,s)=\{e\in E'\mid \pi_1(e)=v \}$; 
$\delta^{\textup{\texttt{con}}}$ is a transition function assigning to $s$ and action profile $(e_r,e_d)$, s.t. $ e_i\in act(i,s)$ (for $i\in\{r,d\}$) a unique state: 
        $$\delta^{\textup{\texttt{con}}}((E,v),(e,e'))= \begin{cases} 
        (E,v)  & \text{if } e=e';  \\ 
        (E\setminus \{e'\},\pi_2(e)) & \text{otherwise}. \end{cases}$$

\end{definition}


$\mathbb{S}^\mathtt{tb}$ and $\mathbb{S}^\mathtt{con}$  enforce the agents to enact a very strict protocol of moves, with the \textup{\texttt{skip}}-action used in to account for an agent `making way' for the other to move in a turn-based game. We can also use a \textup{\texttt{skip}}-type of action to allow for the active \emph{choice} of `doing nothing'. This would enable us to simulate any order of play: agents taking turns strictly, acting concurrently, but also many other scenarios, for instance: one agent making a couple of moves while the other does nothing, or both agents moving simultaneously for a while, before entering the stage of indefinite inactivity. We will now define such a general sabotage game structure. 
\begin{definition}\label{def:generalSabotageGameStructure}
A \emph{general \textup{(\texttt{gen}-)}sabotage games structure} based on a graph $G=(V,E)$ is $\mathbb{S}^{\textup{\texttt{gen}}}(G)=(\mathbb{AG},Act^{\textup{\texttt{skip}}},S^{\textup{\texttt{gen}}},act^\textup{\texttt{gen}},\delta^\textup{\texttt{gen}}),$ where $S^{\textup{\texttt{gen}}}=\mathcal{S}(G)$; for any agent $i\in \mathbb{AG}$ and game-state $s\in S^{\textup{\texttt{gen}}}$, $act^\textup{\texttt{gen}}(i,s)=act^{\textup{\texttt{con}}}(i,s)\cup\{\textup{\texttt{skip}}\}$; and the transition function $\delta^\textup{\texttt{gen}}$ is defined as follows:
    $$\delta^\textup{\texttt{gen}}((E,v),(e,e'))= \begin{cases} 
    (E,v)  & \text{if } e=e'; \\
        (E\setminus \{e'\},v) & \text{if } e\neq e' \text{and } e=\textup{\texttt{skip}};  \\ 
        (E,\pi_2(e)) & \text{if }  e\neq e' \text{and } e'=\textup{\texttt{skip}}; \\
        (E\setminus \{e'\},\pi_2(e)) & \text{otherwise}. \end{cases}$$   
\end{definition}
\noindent In $\mathbb{S}^{\mathtt{gen}}$, at each game-state $s$ each agent can choose to do noting or to make a move. That joint choice will result in either the game progressing to a new state $t$, in which the position or the set of edges (or both) are updated, or the game remaining in $s$ (if the agents choose the same edge or both decide to skip).\footnote{Illustrations of the three kinds of sabotage game structures are given in Appendix B.}

All of the above sabotage games structures assume the usual asymmetry between the powers of the players: demons' choices are global (demon can pick any edge in the graph at any round), and runner's actions are local (runner can only pick an edge sourced at the runner's current position). In our setting, this is rendered in the following way: for any ${\textup{\texttt{x}}}\in \{\mathtt{tb},\mathtt{con},\mathtt{gen}\}$, and any $s\in S^{\textup{\texttt{x}}}$, $act^{\textup{\texttt{x}}}(r,s)\subseteq act^{\textup{\texttt{x}}}(d,s)$.

\begin{definition}
Given ${\textup{\texttt{x}}}\in \{\mathtt{tb},\mathtt{con},\mathtt{gen}\}$, a play $\lambda$ in $\mathbb{S}^{\textup{\texttt{x}}}(G)$ is a maximal sequence of game-states $s_0,s_1,\ldots $, such that for all $i\geq 0$, $s_i$ is a $\delta^{\textup{\texttt{x}}}$-successor of $s_{i-1}$ in $\mathbb{S}^{\textup{\texttt{x}}}(G)$; $\lambda[i]$ is the $i$-th state in $\lambda$; $\lambda[i,j]$ is the finite segment $\lambda[i],\ldots,\lambda[j]$, which is also called a `history' and denoted by $h$ (with $last[h]$ being the last element in $h$, and $length[h]$ being the length of $h$);\footnote{Note that the $length$ will be also applied to finite plays.} $\lambda[i,\infty]$ is the suffix $\lambda[i],\lambda[i+1],\ldots$; if $\lambda[0]=s$, $\lambda$ is called $s$-play. The set of all plays of $\mathbb{S}^{\textup{\texttt{x}}}(G)$ will be denoted by $Plays(\mathbb{S}^{\textup{\texttt{x}}}(G))$, and the set of all $s$-plays of $\mathbb{S}^{\textup{\texttt{x}}}(G)$ by $Plays(\mathbb{S}^{\textup{\texttt{x}}}(G),s)$. 
\end{definition}
\begin{proposition}
Let $G=(V,E)$ be a graph. (1.)~$Plays(\mathbb{S}^{\textup{\texttt{tb}}}(G))$ is a finite set of plays of finite length. 
(2.)~For any $\lambda \in Plays(\mathbb{S}^{\textup{\texttt{con}}}(G))$, there is a $\lambda'\in Plays(\mathbb{S}^{\textup{\texttt{gen}}}(G))$, s.t. for all $i< length(\lambda)$, $\lambda[i]= \lambda'[i]$.
(3.)~$Plays(\mathbb{S}^{\textup{\texttt{con}}}(G))$ and $Plays(\mathbb{S}^{\textup{\texttt{gen}}}(G))$ are infinite and contain infinite plays.
\end{proposition}


\paragraph{Structural labeling functions and computations}
Labeling functions assign propositions to game-states in sabotage game structures. Given a graph $G=(V,E)$, we define the set of graph propositions which includes one dedicated proposition for each vertex, and one proposition for each edge in $G$, i.e., $\Phi^{G}=\Phi^V \cup {\Phi^E}$, where $\Phi^V=\{p_i\mid v_i\in V\}$ and $\Phi^E=\{q_i\mid e_i\in E\}$. Intuitively, the elements of $\Phi^{G}$ will account for the position of the runner and the remaining edges at a given game-state. 
\begin{definition}\label{def:struc_label}
A \emph{structural labeling function} $\mathbb{L}: \mathcal{S}(G)\rightarrow \mathcal{P}(\Phi^{G})$ is defined as $\mathbb{L}((E',v))=\{q_{i}\mid e_i\in E'\}\cup \{p_i\mid v=v_i\}$. A \emph{structural computation} of a sabotage game structure $\mathbb{S}(G)$ is a sequence $\mathbb{L}(\lambda[0])$, $\mathbb{L}(\lambda[1]),\ldots$ such that $\lambda\in Plays(\mathbb{S}(G))$. 
\end{definition}

\begin{proposition}
(1.)~For any $\lambda\in Plays(\mathbb{S}^{\textup{\texttt{tb}}}(G),((E,v),r))$ and an odd $i< length(\lambda)$, $\mathbb{L}(\lambda[i])\neq \mathbb{L}(\lambda[i+1])$. 
(2.)~For any $\lambda \in Plays(\mathbb{S}^{\textup{\texttt{tb}}}(G))$, there is a $\lambda'\in Plays(\mathbb{S}^{\textup{\texttt{gen}}}(G))$, such that for all $i< length(\lambda)$, $\mathbb{L}(\lambda[i])= \mathbb{L}(\lambda'[i])$. 

\end{proposition}

While the set of turn-based computations and the set of concurrent computations are in general incomparable, we can easily see that any (different from the predecessor) game-state reachable in the concurrent game structure is reachable in the turn-based game, only it takes two steps instead of one.

\begin{example}\label{ex:sabotageGame}
Consider a graph $G=(V,E)$ such that $V=\{0,1,2\}$ and  $E=\{(0,1),(0,2),(1,2)\}$, and the runner starts in $0$, see Fig.~\ref{fig:sabotageGame}a. Runner and demon now choose an edge each---if this joint action contains the same edge twice (runner moves on the edge the demon attempts to delete), it will have no effect. Fig.~\ref{fig:sabotageGame}b, shows the effect of a harmonious choice: runner moves from $0$ to $1$, while demon deletes $(0,2)$. Next, runner moves to $2$ and demon deletes the edge $(0,1)$, Fig.~\ref{fig:sabotageGame}c. The the runner is now stuck in $v_2$ so he has no executable action and the concurrent play ends. This play is not a play in turn-based structure, but it can be transformed into one by `unraveling' every concurrent action profile $(e_r,e_d)$ into a sequence of action profiles in a turn-based play: $(e_r, \texttt{skip}),(\texttt{skip},e_d)$. 
    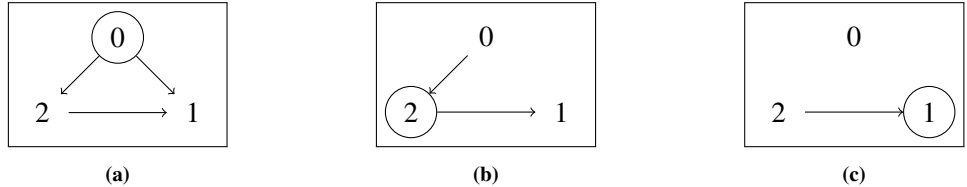
\begin{figure}[H]
        \centering

         \begin{center}
         \begin{subfigure}[b]{0.3\textwidth}
            \begin{center}
\fbox{\begin{tikzpicture}[auto,node distance=1cm]
\node[circle, draw] (0) at (1,1) {$0$};
\node[circle] (2) at (0,0){$2$};
\node[circle] (1) at (2,0) {$1$};
\path    (0) edge [->] node [] {} (2);
\path (2) edge[->] node {} (1);
\path    (0) edge [->] node [] {} (1);       
\end{tikzpicture}}
\caption{}
\end{center}
\end{subfigure}
\begin{subfigure}[b]{0.3\textwidth}
   \begin{center}
\fbox{\begin{tikzpicture}[auto,node distance=1cm]
\node[circle] (0) at (1,1) {$0$};
\node[circle, draw] (2) at (0,0){$2$};
\node[circle] (1) at (2,0) {$1$};
\path    (0) edge [->] node [] {} (2);
\path (2) edge[->] node {} (1);
\end{tikzpicture}}
\caption{}
\end{center}
\end{subfigure}
\begin{subfigure}[b]{0.3\textwidth}
\begin{center}
\fbox{\begin{tikzpicture}[auto,node distance=1cm]
\node[circle] (0) at (1,1) {$0$};
\node[circle] (2) at (0,0){$2$};
\node[circle, draw] (1) at (2,0) {$1$};
\path (2) edge[->] node {} (1);
\end{tikzpicture}}
\caption{}
\end{center}
\end{subfigure}
\end{center}
\vspace{-0.5cm}
        \caption{A play of a concurrent sabotage game structure. The circle marks the runner's position in the graph, the consecutive game-states (b) and (c) are produced by a joint action of runner and demon executed in the previous state.}
    \label{fig:sabotageGame}
    \end{figure}
\end{example}

\paragraph{Sabotage game models}
Labeling functions turn our sabotage game structures into sabotage game models, which will in turn allow interpreting the language of ATL$^\ast$.  
\begin{definition}

\item Let $(G,v_0,b)$ be an LSG, and $\textup{\texttt{x}}\in\{\textup{\texttt{tb}},\textup{\texttt{con}},\textup{\texttt{gen}}\}$. A \emph{liveness $\textup{\texttt{x}}$-sabotage game model} is $\mathbb{M}^{L\textup{\texttt{x}}}(G)=(\mathbb{S}^\textup{\texttt{x}}(G),\mathbb{L})$, where $\mathbb{S}^\textup{\texttt{x}}(G)$ is an $\textup{\texttt{x}}$-sabotage game structure based on $G$ and $\mathbb{L}: S^\textup{\texttt{x}}\rightarrow \mathcal{P}(\Phi^{G})$ is the structural labeling function, as in Def.~\ref{def:struc_label}.~
\item Let $(G,v_0,v_g)$ be an RSG, and $\textup{\texttt{x}}\in\{\textup{\texttt{tb}},\textup{\texttt{con}},\textup{\texttt{gen}}\}$. A \emph{reachability $\textup{\texttt{x}}$-sabotage game model} is $\mathbb{M}^{R\textup{\texttt{x}}}(G,v_g)=(\mathbb{S}^\textup{\texttt{x}}(G),\mathbb{L}^g)$, where $\mathbb{S}^\textup{\texttt{x}}(G)$ is an $\textup{\texttt{x}}$-sabotage game structure and $\mathbb{L}^g: S^\textup{\texttt{x}}\rightarrow \mathcal{P}(\Phi^{G}\cup\{g\})$, such that for any $s\in S^\textup{\texttt{x}}$, $\mathbb{L}^g(s)=\mathbb{L}(s)\cup\{g\}$ if $v_g=v$, and $\mathbb{L}^g(s)=\mathbb{L}(s)$ otherwise.\footnote{When $G$ and $v_g$ is clear from context, we will write $\mathbb{M}^{L\textup{\texttt{x}}}$ instead of $\mathbb{M}^{L\textup{\texttt{x}}}(G)$, and $\mathbb{M}^{R\textup{\texttt{x}}}$ instead of $\mathbb{M}^{R\textup{\texttt{x}}}(G,v_g)$.}

\end{definition}

\paragraph{Alternating-time Temporal Logic, ATL$^\ast$}
Let us recall the setting of ATL$^\ast$ (following the standard approach in \cite{Demri:2016}). We will interpret its language on sabotage game models.
\begin{definition}[Syntax of ATL$^\ast$]
Let $\Phi$ be a set of propositions, $p\in\Phi$, and let $\mathbb{AG}$ be a set of agents with $C\subseteq \mathbb{AG}$. The syntax of ATL$^\ast$ is given by:
$
\varphi::= \top \, | \, p \,| \,\neg \varphi \,| \, \varphi \land \varphi\,
| \, X\varphi\, | \,G \varphi \,| \,\varphi U \varphi\,| \, \langle\!\langle C\rangle\!\rangle \varphi$. 
We also define $F\varphi:=\top U\varphi$, $\llbracket C \rrbracket \varphi := \neg \langle\!\langle C\rangle\!\rangle \neg \varphi$, and $\bot$ as $\neg\top$.    
\end{definition}

We will interpret the language ATL$^\ast$ on sabotage models, using the standard ATL$^\ast$ semantics.  The formulae are interpreted on states and plays of a sabotage game model. Accordingly, a distinction between state and path formulae is made, with state formulae given by: 
$\varphi::= \top \, | \, p \,| \,\neg \varphi \,| \, \varphi \land \varphi\,
| \, \langle\!\langle C\rangle\!\rangle \gamma$, 
and path formulae by:
$ \gamma::= \varphi \,| \,\neg \gamma \,| \,\gamma\land \gamma\,
| \, X\gamma\, | \,G \gamma\,| \gamma U \gamma$.

\begin{definition}
A positional strategy in $\mathbb{S}^\textup{\texttt{x}}$ of $a\in\mathbb{AG}$ is $str_a:S^\textup{\texttt{x}}\rightarrow Act^\texttt{skip} $, s.t. $str_a(s)\in act^{\textup{\texttt{tb}}}(a,s)$. 
We will refer to any $C\subseteq \mathbb{AG}$ as \emph{coalition} $C$. A strategy in $\mathbb{S}^{\textup{\texttt{x}}}$ for a coalition $C$ is a tuple of strategies, one for every $a\in C$.  The set of executable action profiles at $s$ is $act^{\textup{\texttt{x}}}(s)=\Pi_{a\in\mathbb{AG}}act^{\textup{\texttt{x}}}(a,s)$. A $C$-action is a tuple $\alpha_C$ such that $\alpha_C(a)\in act^{\textup{\texttt{x}}}(a,s)$ for every $a\in C$, and $\alpha_C(a')= \#_{a'}$ for every $a'\notin C$. An action profile $\alpha\in  act^{\textup{\texttt{x}}}(s)$ \emph{extends} a $C$-action $\alpha_C$, denoted $\alpha_C \sqsubseteq \alpha$, if $\alpha(a)=\alpha_C(a)$ for every $a \in C$. \emph{The outcome set of the $C$-action}, $\alpha_C$ at $s$ in $\mathbb{S}^{\textup{\texttt{x}}}$ is the set of states $post_{\mathbb{S}^{\textup{\texttt{x}}}} (s, \alpha_C) := \{\delta^{\textup{\texttt{x}}}(s, \alpha) \mid \alpha \in act^{\textup{\texttt{x}}}(s) \text{ and } \alpha_C \sqsubseteq \alpha\}$.
Finally, the set of plays that can be realized by $C$ when following the strategies in $\mathtt{str}_C$ is defined as:
$Plays(\mathbb{S}^{\textup{\texttt{x}}}(G),s,\mathtt{str}_C):= \{\lambda \in Plays(\mathbb{S}^{\textup{\texttt{x}}}(G),s)\mid \lambda[j+1] \in  post_{\mathbb{S}^{\textup{\texttt{x}}}} (\lambda[j],\mathtt{str}_C) \text{ for all } j <length(\lambda)\}$.
\end{definition}
\begin{definition}\label{def:atlSemantics}
Let $\mathbb{M}^{\textup{\texttt{y}\texttt{x}}}$ be a sabotage game model with $\textup{\texttt{y}}\in\{\texttt{R},\texttt{L}\}$ and $\textup{\texttt{x}}\in\{\texttt{tb},\texttt{con},\texttt{gen}\}$, $s$ a game-state in $S^{\textup{\texttt{x}}}$, $\lambda \in Plays(\mathbb{S}^{\textup{\texttt{x}}}(G))$ and let $\Phi$ be a set of propositions with $p\in\Phi$. 
The semantics ATL$^\ast$ is defined inductively for the state and path formulae in the following way:\end{definition}

\begin{center}
\begin{small}
\setlength\tabcolsep{3pt}
\noindent\begin{tabular}[t]{|lcl|}\hline
State formulae: & & \\\hline
$\mathbb{M}^{\textup{\texttt{y}\texttt{x}}},s   \models p $&iff  & $p\in \mathbb{L}^g(s)$, for all $p \in \Phi$ \\
$\mathbb{M}^{\textup{\texttt{y}\texttt{x}}},s   \models \neg \varphi$ & iff & $\mathbb{M}^{\textup{\texttt{y}\texttt{x}}},s \not \models \varphi$\\
$\mathbb{M}^{\textup{\texttt{y}\texttt{x}}},s   \models  \varphi \land \psi$ & iff & $\mathbb{M}^{\textup{\texttt{y}\texttt{x}}},s  \models \varphi$ and $\mathbb{M}^{\textup{\texttt{y}\texttt{x}}},s \models \psi$\\
$\mathbb{M}^{\textup{\texttt{y}\texttt{x}}},s   \models \langle\!\langle C\rangle\!\rangle \varphi$ & iff & there is a $C$-strategy  $\mathtt{str}_C$ \\
& & s.t. $\mathbb{M}^{\textup{\texttt{y}\texttt{x}}},\lambda \models \varphi$ holds for all \\
   & &  $\lambda\in Plays(\mathbb{S}^{\textup{\texttt{x}}}(G),s,\mathtt{str}_C) $\\\hline
\end{tabular}\vspace{0.7cm}
\setlength\tabcolsep{3pt}
\begin{tabular}[t]{|lcl|}\hline
Path formulae: & & \\\hline
   $\mathbb{M}^{\textup{\texttt{y}\texttt{x}}},\lambda   \models \varphi $ & iff &  $\mathbb{M}^{\textup{\texttt{y}\texttt{x}}},\lambda[0]\models \varphi$ for every \\
   && state formula  $\varphi$\\
      $\mathbb{M}^{\textup{\texttt{y}\texttt{x}}},\lambda   \models \neg \psi$   &     iff  & $\mathbb{M}^{\textup{\texttt{y}\texttt{x}}},\lambda \not \models \psi$\\
              $\mathbb{M}^{\textup{\texttt{y}\texttt{x}}},\lambda   \models  \psi_1 \land \psi_2$   &     iff &  $\mathbb{M}^{\textup{\texttt{y}\texttt{x}}},\lambda \models \psi_1$ and $\mathbb{M}^{\textup{\texttt{y}\texttt{x}}},\lambda \models \psi_2$\\
  $\mathbb{M}^{\textup{\texttt{y}\texttt{x}}},\lambda  \models X \psi$   &     iff   & $\mathbb{M}^{\textup{\texttt{y}\texttt{x}}},\lambda[1,\infty)\models \psi$\\
   $\mathbb{M}^{\textup{\texttt{y}\texttt{x}}},\lambda  \models  G \psi$   &     iff  &  $\mathbb{M}^{\textup{\texttt{y}\texttt{x}}},\lambda[i,\infty)\models \psi$   for all $i\geq 0$\\
    $\mathbb{M}^{\textup{\texttt{y}\texttt{x}}},\lambda  \models  \psi_1 U \psi_2$   &     iff  &    there is a position  $i\geq 0$  \\
    & & s.t.   $\mathbb{M}^{\textup{\texttt{y}\texttt{x}}},\lambda[i,\infty)\models \psi_2$   and   \\
    & & $\mathbb{M}^{\textup{\texttt{y}\texttt{x}}},\lambda[j,\infty)\models \psi_1$   for all  $0 \leq j<i $ \\\hline
\end{tabular}
\end{small}
\end{center}

\paragraph{Existence of winning strategies in Reachability Sabotage Games}
When considering winning conditions in reachability sabotage games, the plays of interest are those in which the goal vertex is eventually reached, i.e., at some point in the play arrives at a game-state in which $g$ true. In the initial game-states of such plays $\top U g$ holds (or equivalently, using the future operator, $Fg$). The existence of a winning strategy for the runner in an RSG can be then expressed by: $\langle\!\langle \{r\} \rangle\!\rangle F g$, which says that runner can enforce that (has a strategy such that on all plays in that strategy) $g$ will eventually hold. Let us first consider \emph{turn-based} reachability sabotage games.  

\begin{proposition}\label{prop:tb-winning-RSG}
Let $G=(V,E)$ be a graph. Runner has a winning strategy in a turn-based RSG $(G,v_0,v_g)$ \ iff \ $\mathbb{M}^{R\textup{\texttt{tb}}},((E,v_0),r)\models \langle\!\langle \{r\} \rangle\!\rangle F g$.
\end{proposition}

The classical winning conditions of the RSG assume specific incentives of the two players. Motivated by modeling the interaction between teacher (demon) and learner (runner) in learning scenarios, \cite{Gierasimczuk:2009} introduced variations of the classical winning conditions of the sabotage game: runner could be eager or unwilling to get to the goal and demon could be helpful or unhelpful in the process. Accordingly, we get four different turn-based sabotage games: the classical $RSG^{EU}$ (eager runner and unhelpful demon), $RSG^{EH}$ (eager runner and helpful demon), $RSG^{UH}$ (unwilling runner and helpful demon), and $RSG^{UU}$ (unhelpful demon and unwilling runner). Tab.~\ref{tab:comparison_SML_ATL} shows the ATL formulae expressing winning conditions in these games (in the sense of Prop.~\ref{prop:tb-winning-RSG}), and compare them to way they can be expressed in SML. 
\begin{table}[H]
\begin{center}
\setlength\tabcolsep{7pt}
\renewcommand{\arraystretch}{1.1}
\noindent\begin{tabular}[t]{|l|l|c|l|}\hline
Type of RSG & Implied winner & SML formula ($\rho_k$, $k=|E|$) & ATL formula \\\hline
$RSG^{EU}$ & runner& $\rho_0:=g, \ \rho_{n+1}:= g\vee \lozenge\blacksquare\rho_n$ \cite{Benthem:2005}& $\langle\!\langle \{r\} \rangle\!\rangle F g$ \\
$RSG^{EH}$& runner and demon & $\rho_0:=g, \ \rho_{n+1}:= g\vee  \lozenge\blacklozenge\rho_n$ \cite{Gierasimczuk:2009} & $\langle\!\langle \{r,d\} \rangle\!\rangle  F g$\\
$RSG^{UH}$ & demon & $\rho_0:=g, \ \rho_{n+1}:= g\vee (\lozenge \top \wedge \square\blacklozenge\rho_n)$ \cite{Gierasimczuk:2009} &  $\langle\!\langle \{d\} \rangle\!\rangle  F g$\\
$RSG^{UU}$ & runner and demon & $\rho_0:=\neg g, \ \rho_{n+1}:= \neg g\vee  \lozenge\blacklozenge\rho_n$ & $\langle\!\langle \{r,d\} \rangle\!\rangle  \neg Fg$\\\hline
\end{tabular}
\end{center}

\caption{Comparison of SML and ATL formulae expressing winning conditions in sabotage games with various incentives}\label{tab:comparison_SML_ATL}
\end{table}
\noindent Additionally, note that the winning condition for demon in the standard $RSG^{EU}$ can be expressed by $\langle\!\langle \{d\} \rangle\!\rangle  G \neg g$ (demon has a strategy such that in all resulting plays the runner does not eventually reach the goal). Note that in turn-based sabotage games this is equivalent to $\llbracket \{r\} \rrbracket  G \neg g$, which expresses that runner cannot avoid losing. This duality property is an instance of the fact that turn-based zero-sum games are \emph{determined} (see~\cite{Buchi:1969aa,Gurevich:1982aa,Alur:2002}). In general, any ATL formulae $ \langle\!\langle C \rangle\!\rangle \varphi $ in Tab.~\ref{tab:comparison_SML_ATL} is equivalent to $ \llbracket \mathbb{AG}\setminus C \rrbracket\neg \varphi $, e.g., the ATL formula for $RSG^{EH}$ is equivalent to $\llbracket \emptyset \rrbracket  F g$, 
stating the CTL-expressible condition that there is a play on which runner reaches the goal. 

As a matter of fact we can show that runner can (almost) never win a turn-based RSG.
\begin{proposition}\label{prop:tb-demonwinning}
Let $((V,E),v_0,v_g)$ be a RSG with $v_0\neq v_g$ and $(v_0,v_g)\notin E$. Then we have that $\mathbb{M}^{R\textup{\texttt{tb}}},((E,v_0),r)\not\models \langle\!\langle \{r\} \rangle\!\rangle F g$.
\end{proposition}
\noindent In other words, unless runner begins in the goal vertex, or the goal vertex is reachable from the start vertex in one step, demon will always be able to prevent runner from reaching the goal. The argument goes as follows. Assume that up to and including the point $i$ in the play the goal has not been reached, and that $\lambda[i]=((E',v),d)$ (it is demon's turn to move). Demon's winning strategy is: if there is an $e\in E'$ s.t. $e=(v,v_g)$, then play $e$; otherwise play any $e'\in E'$. The strategy works because in a simple graph there can only be one edge between any two vertices. This handicap of the runner is why sabotage games are often studied in the context of \emph{multi-graphs}, which allow multiple edges between any two vertices. Our temporal setting can be easily extended to that setup, by extending the set of actions with \emph{indexed} pairs of vertices, see e.g., \cite{Gierasimczuk:2009}. In this paper we decided to keep to simple (and directed) graphs for ease of exposition.

\smallskip

Switching to \emph{concurrent games} does not change the winning conditions for different types of RSG listed in Tab.~\ref{tab:comparison_SML_ATL}. It does however, further limit the runner's power. In concurrent games any action taken by runner can be in principle `canceled' by demon choosing the same edge, preventing it from ever leaving the initial game-state. 
Moreover, note that even though both turn-based and concurrent sabotage games afford cooperative strategies, they have different scopes in the two types. First, the following proposition expresses that if there is a joint winning strategy in the concurrent version of the game, then there is one in the turn-based version of the same game.
\begin{proposition}
Let $((V,E),v_0,v_g)$ be a RSG. If  $\mathbb{M}^{R\textup{\texttt{con}}},(E,v_0)\models \langle\!\langle \{r,d\} \rangle\!\rangle  F g$, then $\mathbb{M}^{R\textup{\texttt{tb}}},((E,v_0),r)\models \langle\!\langle \{r,d\} \rangle\!\rangle  F g$.
\end{proposition}
\noindent The opposite implication however does not hold. To see this, consider the following example.

\begin{example}\label{ex:concurrent_infinite}
In the graph $G^{\infty}=(V,E)$ with $V=\{v_0,v_g\}$ and $E=\{(v_0,v_g)\}$ there is a winning strategy for the coalition of runner and demon in a turn-based sabotage game, but there is no such strategy in the concurrent game played on $G^\infty$. In the concurrent play, the only joint choice for the players is to select the same edge $(v_0,v_g)$ at the start of the game, and that will keep them from ever reaching $v_g$.
\end{example}

Example \ref{ex:concurrent_infinite} also demonstrates that in general neither demon nor runner, and not even the coalition of the two players, is guaranteed to be able to force a concurrent game to end. In other words, in RSG $(G^\infty,v_0,v_g)$ we have that $\mathbb{M}^{R\textup{\texttt{con}}},(E,v_0)\not\models \langle\!\langle \{r,d\} \rangle\!\rangle  F X\bot.$


\paragraph{Existence of winning strategies in Liveness Sabotage Games}
The goal in liveness sabotage isn't to reach the goal state. Instead, we want runner to stay `alive' for (at least) a given number of moves. While the existence of a possible move that a runner can make can be expressed by the formula $X\top$ being true in a game-state controlled by runner, we want it to be at a \emph{specific time-point}, after each player moved $b$-times. In order to be able to do this we extend the language of ATL$^\ast$ with the parametrized until-operator borrowed from metric temporal logic \cite{Koymans:1990aa}, allowing a new path-formula, $\psi U_i \varphi$ for $i\in\mathbb{N}$ with the following meaning: $\text{$\mathbb{M},\lambda  \models  \psi U_i \varphi$   iff   $\mathbb{M},\lambda[i,\infty)\models \varphi$ and  for all $j$ s.t. $0 \leq j<i $, $\mathbb{M},\lambda[j,\infty)\models \psi$}.$ 
It allows expressing the condition of runner having `somewhere to go' for $i$ or its rounds, and so that the coalition of players $C$ has a strategy to make the game last at least $i$-rounds. In the case of turn-based sabotage games the following formula will suffice:
$\langle\!\langle \{r\} \rangle\!\rangle \top  U_{2b}  (X\top)$, since at
    even rounds of the turn-based RSG it is always runner's turn. 
    \begin{proposition}
      Let $G=(V,E)$ be a graph. Runner has a winning strategy in a turn-based LSG $(G,v_0,b)$ \ iff \ $\mathbb{M}^{R\textup{\texttt{tb}}},((E,v_0),r)\models \langle\!\langle \{r\} \rangle\!\rangle \top  U_{2b}  (X\top)$.  
    \end{proposition}
\noindent Note that the above formula also implies that $\langle\!\langle \{r\} \rangle\!\rangle (X\top) U_{2b}  (X\top)$, i.e., that runner can enforce the play $\lambda$, s.t. for all $i\leq 2b$, $\lambda[i]$ has a $\delta^\textup{\texttt{tb}}$-successor.


In concurrent LSGs $((V,E),v_0,b)$, as long as $E\neq\emptyset$, runner and demon can jointly enforce that the game is live forever, i.e., $\mathbb{M}^{L\textup{\texttt{con}}},(E,v_0)\models\langle\!\langle \{r,d\} \rangle\!\rangle GX\top$, by indefinitely applying $(e_r,e_d)$, s.t. $e_r=e_d$. A somewhat surprising observation is that what constitutes a problem for runner in an RSG, is an advantage in LSG. Namely, as long as demon chooses the same edge as the runner, runner remains alive, i.e., she retains a position with a successor. In RSG, the canceling choice prevents the runner from progressing towards the goal. 





\paragraph{Sabotage games and the minimum cut of a graph}

A strategy to destroy the edge-connectivity of a graph could be to find and sever the \emph{minimum cut} of the graph. The problem of finding minimum cuts in graphs is well-studied in the field of graph theory (see, e.g., \cite{ford-fulkerson-1956,dinitz-karzanov-lomonosov-1976,karger-1993}). Finding such cuts is often relevant with respect to two specific vertices of the graph, in such a case the problem is called a minimal \emph{$s{-}t$ cut}.

\begin{definition}[Minimum cut \cite{Dantzig:1956}, see also \cite{kleinbergTardos2013}] The \emph{minimum cut} of a graph $G=(V,E)$ is the minimum number of edges in $E$ that, when removed from the graph, partition the vertices into two disjoint sets, $V'$ and $V''$, such that no vertex from $V'$ is reachable from $V''$. The \emph{minimum-cut of a  graph $G=(V,E)$ with respect to the vertices $s,t\in V$} is the minimum number of edges in $E$ that after their removal from $G$ $s$ is not reachable from $t$.\end{definition}

The minimum $s-t$ cut has been characterized in terms of paths in the graph in the following way.

\begin{theorem}[Menger's theorem \cite{Menger:1927aa}]
Let $G=(V,E)$ and let $x, y\in V$ with $x\neq y$. The size of the minimum cut for $x$ and $y$ (the minimum number of edges whose removal disconnects $x$ and $y$) is equal to the maximum number of pairwise edge-independent paths from $x$ to $y$.
\end{theorem}

Such minimum $s-t$ cuts are useful when working with static $s$ which can be seen as unmoving runner. If however the runner moves (as she should), in some cases she can `escape' the execution of a static min-cut before demon manages to complete it. So, a question of \emph{dynamic minimum $s{-}t$ cut} appears: What is the minimal number of rounds demon must play in order to prevent the runner to be able to reach the goal? To see that the standard notion of (static) minimum $s{-}t$ cut and the sabotage-based notion of the dynamic minimum $s{-}t$ cut are different, consider the following example.
 
 \begin{example}
Let us play a turn-based RSG $(G,s,t)$ on the graph $G$ depicted below. Firstly, note that the minimum static $s{-}t$ cut of the graph is of size $2$ and contains the edges $\{(s,u),(s,w)\}$. A naive demon could think that removing these edges one by one would allow him to win the game in two moves. Unfortunately, runner is allowed to move in-between the demon's moves, so she would manage to escape the static min-cut by either moving to $u$ or $w$ in its first move. The min-cut from both $u$ or $w$ to $t$ is of size $3$, so the demon has underestimated the number of edges to be deleted, and must reconsider which edge should be cut. 
There are in fact three dynamic minimal $s{-}t$ cuts of the graph (depending on how the runner moves): $\{\{(s,u),(u,v),(u,t)\}, \{(s,w),(w,v),(w,t)\},\{(v,t),(u,t),(w,t)\}\}$ and the demon cannot expect to win in less than three rounds. 
    
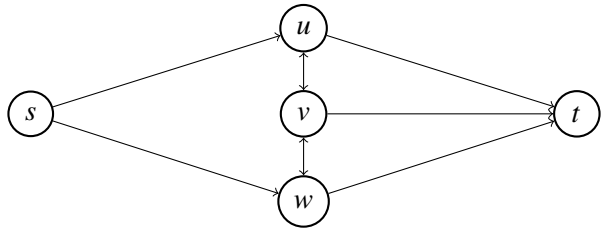
\begin{center}
    
\begin{tikzpicture}[
    node distance=.5cm,
    every node/.style={circle, draw, thick, minimum size=5mm},
    every edge/.style={draw, thick}
]

\node (v) {$v$};
\node (u) [above=of v] {$u$};
\node (w) [below=of v] {$w$};
\node (s) [left=3cm of v] {$s$};
\node (t) [right=3cm of v] {$t$};

\draw[->] (s) -- (u);
\draw[->] (s) -- (w);

\draw[<->] (u) -- (v);
\draw[<->] (v) -- (w);

\draw[->] (u) -- (t);
\draw[->] (v) -- (t);
\draw[->] (w) -- (t);

\end{tikzpicture}
    \end{center}
\end{example}

The set of edges mentioned last in the enumeration in the Example above is of particular importance. Unlike in the setting of multi-graphs, which are often studied in the context of sabotage games, in standard graphs demon always has a winning strategy in the sabotage game: remove an arbitrary edge (or do nothing) until the runner is at some $v$ that is one step away from the goal vertex $v_g$, then remove $(v,v_g)$; and repeat that procedure as long as the runner's position is disconnected from the goal. To minimize the number of removals that disconnect the goal vertex from the rest of the graph, it is enough if demon focuses on only the edges that end at $v_g$. Even in that case however, the order of removals makes a difference, and must depend on the position of the runner.

To capture the dynamic $s{-}t$ cuts ($v_0{-}v_g$ cuts, in our case) we are interested in those $v_0$-plays that validate the formula $F \langle\!\langle \emptyset \rangle\!\rangle G\neg g$, which states that there is a point at which in all futures the goal will never be reached (equivalently stated as $\top U (\langle\!\langle \emptyset \rangle\!\rangle G\neg g$)). With the parametrized until-operator we can put a time-stamp on that moment, and require that it is obtained by the demon as soon as possible. Demon's strategy (set of plays) corresponding to the minimal dynamic $v_0{-}v_g$ cut then can be expressed in the following way.
\begin{definition}
The minimal $\textup{\texttt{tb}}$-dynamic $v_0{-}v_0$ cut of the graph $G=(V,E)$ is $k$ iff $v_0=s$ and $v_g=t$ and $k$ is the smallest such that
$\mathbb{M}^{R\textup{\texttt{tb}}},((E,v_0),r)\models \langle\!\langle \{d\} \rangle\!\rangle (\top U_k (\langle\!\langle \emptyset \rangle\!\rangle G\neg g))$.
\end{definition}
Let us mention, that a number of classical results address variants of the minimum $s-t$ cut problem where the source (or the pair $s,t$) is not fixed in advance. The classical solution is the \emph{Gomory--Hu tree} \cite{gomory1961multi}, which represents all pairwise minimum $s-t$ cuts of an undirected graph using only $n-1$ maximum-flow computations. Later work by Gusfield simplified the construction and analysis of such cut trees \cite{gusfield1990very}. In addition, research on dynamic graph algorithms studies data structures that maintain connectivity and cut-related information while the graph changes or while queries for arbitrary terminal pairs are issued \cite{holm2001poly}. Our ATL$^\ast$-based rendering of Sabotage Games brings logic and the algorithmic graph theory closer, and we plan to deepen this connection in a follow-up work.

\paragraph{Angelic Sabotage Games}



\vspace{-0.2cm}
To make a proper use of the power of ATL$^\ast$ it would be interesting to consider more complex infinite sabotage games. This can be done if apart from deletion of edges, we allow also their addition. In fact, many authors allow also for a `positive' type of dynamics of the graph (see, e.g., \cite{Baltag:2022aa}). \emph{Angelic sabotage games}, apart from runner and demon, also include an agent of positive change---a builder of edges. The following could be one example of such an extension.

\begin{definition}\label{def:angelic_play}
A \emph{strictly turn-based angelic sabotage play} is a (possibly infinite) sequence of game-states $s^0,\ldots, s^n$, such that $s^0=(E^0,v^0)$ is given by $(E,v_0)$, and for any $k$ s.t. $0<k\leq n$: $$s^{k}=
   \begin{cases}
         (E^{k-1},v^{k}),\text{  s.t.  } (v^{k-1},v^{k})\in E^{k-1} & \text{if } k\equiv 1 \pmod{3};\\
    (E^{k-1}\setminus \{(x,y)\},v^{k-1}),\text{ s.t. }(x,y)\in E^{k-1} & \text{if } k\equiv 2 \pmod{3};\\
          (E^{k-1}\cup \{(x,y)\},v^{k-1}),\text{  s.t.  } (x,y)\in V^2\setminus E^{k-1} & \text{if } k\equiv 0 \pmod{3}.\\
   \end{cases}$$
   \end{definition}
The winning conditions for runner and demon in RSG and LSG are as before, while the angel can be most naturally seen as runner's ally, but various coalitions of players with various incentives can be studied, similarly to the approach presented in Tab.~\ref{tab:comparison_SML_ATL}. The definitions of the $\textup{\texttt{tb}}$, $\textup{\texttt{con}}$, and $\textup{\texttt{gen}}$ angelic sabotage game structures are straight-forward extensions of the ones in Def.~\ref{def:TurnSabotageGameStructure},\ref{def:concurrentSabotageGameStructure}, and \ref{def:generalSabotageGameStructure}. The general difference with the non-angelic kind of sabotage game structures is that the play can recover from breaking the connectivity of the graph and from the runner finding themselves in a dead-end vertex. Since the angelic plays are infinite, we can consider more complex properties, in particular the requirement that a coalition can enforce visiting the goal vertex infinitely many times, $\langle\!\langle C \rangle\!\rangle GFg$, a property that, due to nested temporal operators cannot be expressed in simple ATL, and requires ATL$^\ast$.

\section{Knowledge in Sabotage Games}\label{sec:atel}

To capture cooperative strategic behavior under uncertainty an epistemic extension of ATL was proposed, namely the Alternating-time Temporal Epistemic Logic (ATEL, \cite{Hoek:2003}). It extends the language with knowledge modality $K$ and with the group knowledge modalities $E$ (everybody knows) and $C$ (common knowledge). Combining the strategic operators of ATL with the epistemic operators enables us to express statements about what coalitions of agents can enforce through their joint strategies, while explicitly accounting for what agents know about the game-state and the actions of others. ATEL extends ATL with the following formulae: $K_a \varphi$ ($a$ knows that $\varphi$), $E_\Gamma \varphi$ (everyone in $\Gamma$ knows that $\varphi$) and $C_\Gamma \varphi$ (it's common knowledge among $\Gamma$ that $\varphi$), where $a\in \mathbb{AG}$ and $ \Gamma\subseteq \mathbb{AG}$. 

Interpreting the above formulae in sabotage game models requires enriching them with accessibility relations for the agents. We will refer to such epistemic sabotage game models with $\mathbb{E}^{\textup{\texttt{y}}\textup{\texttt{x}}}$, with $\textup{\texttt{y}}\in\{\textup{\texttt{L}},\textup{\texttt{R}}\}$ and $\textup{\texttt{x}}\in\{\textup{\texttt{tb}},\textup{\texttt{con}},\textup{\texttt{gen}}\}$. The accessibility relations ${\sim_i}\subseteq S^\textup{\texttt{x}}\times S^\textup{\texttt{x}}$, specify for each $i\in\mathbb{AG}$ the scope of their uncertainty in a given game-state, by relating it to those that the agent can not distinguish, given her knowledge, from the current game-state.\footnote{We hence specify the uncertainty of agents to range over game-states (as done e.g. in \cite{Hoek:2003}), rather than on histories (as in \cite{Benthem:2009}). Since classical sabotage games are \emph{history-free}, we can say quite a lot following this assumption.} We will assume the accessibility relations to be equivalence relations. The semantics of the above operators is defined as usual, with the epistemic formulas taken to be state formulas: $\mathbb{M}, s\models K_a\varphi \text{ iff } \text{for all } t, \text{ s.t. }  s \sim_a  t, M, t \models \varphi$; $\mathbb{M}, s\models E_\Gamma\varphi \text{ iff } \text{for all } t, \text{ s.t. }  s \sim_\Gamma^E t, \, M, t \models \varphi$; $\mathbb{M}, s \models C_\Gamma\varphi \text{ iff } \text{for all } t, \text{ s.t. }  s \sim_\Gamma^C t,  M, t \models \varphi$,
where $\sim_\Gamma^E = \bigcup_{a \in \Gamma}\sim_a$, and $\sim_\Gamma^C$ is the transitive closure of $\sim_\Gamma^E $. In the original ATEL \cite{Alur:2002}, the epistemic part was simply added `on top of' the existing ATL semantics, and the definition of the strategies didn't take into account the actual epistemically-based ability to execute them. This simple approach allows talking about some simple aspects of knowledge in sabotage games.

The simplest approach is to assume that the have \emph{perfect information}, i.e., all agents see everything in a state: the edges present and the runner's position. Here we assume that $\sim_i$ for each $i\in \mathbb{AG}$ are identity relations. We can then say, following Prop.~\ref{prop:tb-demonwinning}, that if $((V,E),v_0,v_g)$ is RSG with $v_0\neq v_g$ and $(v_0,v_g)\notin E$, then we have that $\mathbb{E}^{R\textup{\texttt{tb}}},((E,v_0),r)\models E_{\mathbb{AG}}\langle\!\langle \{d\} \rangle\!\rangle G\neg g$. 



\emph{Imperfect information} can be imparted on agents in many different ways, for instance we could assume that they can only see immediate neighbors of the runner's current position, i.e, $(E,v) \sim_i (E',v')$ iff $|E(v)|=|E'(v')|$, where $E(v)=\{e\in E\mid \pi_1(e)=v\}$.\footnote{This specification is somewhat inconsistent with demon's global powers over $G$. This discrepancy could be remedied by restricting the function $act$ in the sabotage game structures, so that $act(d,(E,v))=\{e\in E'\mid \pi_1(e)= v\}$, by making the demon local \cite{Kvasov:2015aa,Kvasov:2016aa,Aucher:2018}).} In such games we an agent might have a winning strategy but she might not know that.
\begin{example}\label{ex:local_symmetric}
Let us assume that a turn-based reachability sabotage game between runner and demon is played on the following graph $G$. 
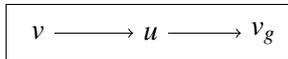
\begin{center}
\fbox{
\begin{tikzpicture}[auto,node distance=1cm]
\node[] (0) at (0,0) {$v$};
\node[] (1) at (1.5,0){$u$};
\node[] (2) at (3,0){$v_g$};
\path    (0) edge [->] node [] {} (1);
\path    (1) edge [->] node [] {} (2);
\end{tikzpicture}}
\end{center}
Clearly, if the starting node is $u$, runner has a winning strategy in RSG: the first move she makes will lead her to the goal vertex. However, given the fact that she can only observe the number of edges leading out of the current node, she can't be sure if she is not in fact in the vertex $v$. If that was the case, she would not be able to win: after moving to $u$, demon can remove the edge $(u,g)$, and the runner loses the game. We then have that  $\mathbb{E}^{R\textup{\texttt{tb}}}(G,v_g),((E,u),r)\models \langle\!\langle \{r\} \rangle\!\rangle F g\wedge \neg K_r \langle\!\langle \{r\} \rangle\!\rangle Fg$.
\end{example}
The above observation might make the reader uneasy---how can we claim something is an agent's strategic range, but they do not know it? This effect is the topic of an important discussion in existing literature on logics of games, which is even more pronounced in the example below.

Let us give runner perfect information, but obscure from demon the location of runner after the first step is made, i.e., for any two $v_0$-plays $\lambda$ and $\lambda'$, and any $i>0$, $\lambda[i]\sim_d \lambda'[i]$, if  $\lambda[i]=((E_1,v),a)$,$\lambda'[i]=((E_2,v'),a)$ and $E_1=E_2$. In this setting, demon might have a winning strategy, but since he cannot know where runner is at a given time, he might not be able to apply it during the play. 

\begin{example}\label{ex:asymmetric}
Consider the graph $G$ below. Runner starts at $v_0$ and both agents know that, and they know the structure of the graph. Runner makes a move hidden to demon. 
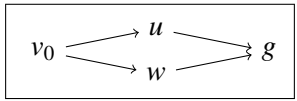
\begin{center}
\fbox{
\begin{tikzpicture}[auto,node distance=1cm]
\node[] (0) at (0,0) {$v_0$};
\node[] (1) at (1.5,0.3){$u$};
\node[] (2) at (1.5,-0.3){$w$};
\node[] (3) at (3,0){$g$};
\path    (0) edge [->] node [] {} (1);
\path    (0) edge [->] node [] {} (2);
\path    (1) edge [->] node [] {} (3);
\path    (2) edge [->] node [] {} (3);
\end{tikzpicture}}
\end{center}
Demon now knows that runner is either in $u$ or in $w$. This is not enough information to apply the strategy that would guarantee winning in the turn-based game with perfect information, i.e., blocker has to chose one of the $(u,g)$ or $(w,g)$, but might get unlucky with that choice, and the runner will still be able to get to $g$. To be able to express this in ATEL, we have to augment the semantics to account for demon's `belief-states' ($\sim_d$ abstraction classes). Such adjustment was first proposed in \cite{Jamroga2003Remarks}, to account for imperfect information strategies. Under this new `imperfect' semantics for the turn-based sabotage game on $G$, we have that $\mathbb{E^\textup{\texttt{Rtb}}}(G,v_0),((E,v_0),r)\not\models_{imp} \langle\!\langle \{d\} \rangle\!\rangle Fg$, while on the original ATEL semantics we would, counterintuitively, get that $\mathbb{E^\textup{\texttt{Rtb}}}(G,v_0),((E,v_0),r)\models K_d\langle\!\langle \{d\} \rangle\!\rangle Fg$. 
\end{example}

\section{Conclusions and future work}\label{sec:conclusion}\label{sec:futureWork}


We examined various kinds of sabotage games from the perspective of ATL. We have studied the classical reachability sabotage game, and introduced liveness sabotage games. Apart from the standard turn-based version, we introduced concurrent sabotage games, which are very natural for the ATL framework. We characterized the existence of winning strategies in these kinds of games. We further related those characterizations to the graph-theoretical problem of minimal $s-t$ cut, and discussed a dynamic version of that problem. We also connected to angelic sabotage games, in which edges can be built. Finally, we have shown how epistemic extensions of ATL can account for knowledge in sabotage games, a need that was highlighted in a recent survey by van Benthem and Liu \cite{vanBenthemLiu2025}.

There are many possible follow-up directions of this work: studying possible extensions of classical sabotage games, such as multiple runners, distributed goals, and infinite sabotage games (especially in the context of angelic games). We are interested in strengthening the relationship with algorithmic graph theory, by linking to various kinds of dynamic min-cut and max-flow problems. Moreover, extending our preliminary epistemic account od Sabotage Games holds special promise. Sabotage Games can be viewed as a natural playground for various notions of strategic ability, e.g., making use of positional and memory-based strategies would allow comparing the power of the players on another level.

\bibliographystyle{eptcs}
\bibliography{bibliography}

\newpage
\section*{Appendix}

\noindent\textbf{Proposition 2.5}
Let $G=(V,E)$ and $b\in \mathbb{N}^+$. Runner has a winning strategy in $LSG=((V,E),v_0,b)$ iff $M((E,v_0)),v_0\models \gamma_b$, with $\gamma_i$ (for $i\in \mathbb{N}$) given by:
  $ \gamma_1:=\lozenge\top, \ \gamma_{n+1}:=\lozenge \blacksquare\gamma_n.$

\begin{proof}
By induction on $b$. Base case, $b=1$. The following are equivalent: 
(1.) runner has a winning strategy in the $LSG=(V,E,v_0,1)$;
(2.) there is a $v_1 \in V$, such that $(v_0,v_1)\in E$ (runner can use it in the present round);
(3.) $M((E,v_0)),v_1\models \top$;
(4.) $M((E,v_0)),v_0\models \lozenge \top$; 
(5.) $M((E,v_0)),v_0\models \gamma_1$.

Inductive hypothesis: runner has a winning strategy in $LSG=(V,E, v_0, n)$ iff $M((E,v_0)),v_0\models \gamma_n$. Take $b=n+1$. 
($\Rightarrow$) Assume that the runner has a winning strategy in LSG $(V,E, v_0, n+1)$. Then there is $v'\in V$, such that $(v,v')\in E$ and for any $(x,y)\in E$, runner has a winning strategy in the LSG $(V,E\setminus\{x,y\},v',n)$. By the inductive hypothesis: for any $(x,y)\in E$, $M((E\setminus\{(x,y)\},v')),v'\models \gamma_n$. Then, by the semantics of $\blacksquare$, $M((E, v_0)),v'\models \blacksquare\gamma_n$. Finally, by our choice of $v'$ and the semantics of $\lozenge$, we get that $M((E, v_0)),v_0\models \lozenge\blacksquare\gamma_n$, i.e., $M((E, v_0)),v_0\models \gamma_{n+1}$. 
($\Leftarrow$) Assume that $M((E, v_0)),v_0\models \gamma_{n+1}$, i.e., $M((E, v_0)),v_0\models \lozenge\blacksquare\gamma_{n}$. Then, there is a $v'\in V$ such that $(v,v')\in E$ and $M((E, v_0)),v'\models \blacksquare\gamma_n$. Then for all $(x,y)\in E$, $M((E\setminus\{(x,y)\}, v_0)),v' \models \gamma_n$, which (by inductive hypothesis) means that for all $(x,y)\in E$, runner has a winning strategy in $(V,E\setminus\{(x,y)\}, v', n)$. But then the runner has a winning strategy in $LSG=(V,E,v_0, n+1)$ by moving first from $v_0$ to $v_1$. 
\end{proof}

\noindent\textbf{Proposition 3.5} 

\noindent(1.)~$Plays(\mathbb{S}^{\textup{\texttt{tb}}}(G))$ is a finite set of plays of finite length. 
\begin{proof}
Take $G=(V,E)$. We show that for any $\lambda\in Plays(\mathbb{S}^{\textup{\texttt{tb}}}(G))$, $length(\lambda)\leq 2|E|-1$. For contradiction assume that there is a $\lambda\in Plays(\mathbb{S}^{\textup{\texttt{tb}}}(G))$, s.t. $length(\lambda)\geq 2|E|$ and $\lambda[2|E|-1]=((E',v),a)$, for some $a\in\mathbb{AG}$. This means that $act^\textup{\texttt{tb}}(\lambda[2|E|-1])\neq \emptyset$, i.e., $E'\neq \emptyset$. But between $\lambda[0]$ and $\lambda[2|E|-1]$ demon must have made $|E|$ moves, i.e., deleted $|E|$ edges, since the game is strictly turn-based, so $E'=\emptyset$. Contradiction. Moreover, $Plays(\mathbb{S}^{\textup{\texttt{tb}}}(G))$ is a finite set because it corresponds to a set of finite (as we've just shown) sequences over a finite alphabet $E$.\end{proof}

\noindent(2.)~For any $\lambda \in Plays(\mathbb{S}^{\textup{\texttt{con}}}(G))$, there is a $\lambda'\in Plays(\mathbb{S}^{\textup{\texttt{gen}}}(G))$, s.t. for all $i\leq length(\lambda)$, $\lambda[i]= \lambda'[i]$.
\begin{proof}
Take a $\lambda\in Plays(\mathbb{S}^{\textup{\texttt{con}}}(G))$. We want to show that $\lambda\in Plays(\mathbb{S}^{\textup{\texttt{gen}}}(G))$. Firstly, $\lambda[0]\in S^\textup{\texttt{gen}}$ because $S^{\textup{\texttt{con}}}=S^{\textup{\texttt{gen}}}$. Secondly, for every $i>0$, $\lambda[i]=\delta^{\textup{\texttt{gen}}}(\lambda[i-1],(e_r,e_d))$ for some $(e_r,e_d)\in act^{\textup{\texttt{gen}}}(\lambda[i-1])$, because for any $s$, $act^{\textup{\texttt{con}}}(s)\subseteq act^{\textup{\texttt{gen}}}(s)$ and for any $s$ and $e,e'\in Act$, $\delta^{\textup{\texttt{con}}}(s,(e,e'))=\delta^{\textup{\texttt{gen}}}(s,(e,e'))$. 
\end{proof}

\noindent(3.)~$Plays(\mathbb{S}^{\textup{\texttt{con}}}(G))$ and $Plays(\mathbb{S}^{\textup{\texttt{gen}}}(G))$ are infinite and contain infinite plays.

\begin{proof}
For contradiction assume that $Plays(\mathbb{S}^{\textup{\texttt{con}}}(G))$ is a finite set, and that its longest play is of length $n\in N$. Take such a longest $\lambda$. Let us construct $\lambda'$ in the following way: $\lambda'[0]:=\lambda[0]$ and for $i\in\{0,\ldots,n-1\}$, $\lambda'[i+1]:=\lambda[i]$. It is clear that $length(\lambda')=n+1$. To see that $\lambda'\in Plays(\mathbb{S}^{\textup{\texttt{con}}}(G))$, consider that $\lambda'[0]\in S^{\textup{\texttt{con}}}$, $\lambda'[1]=\delta^{\textup{\texttt{con}}}(\lambda'[0],(e,e))$ for any choice of $e\in Act$, and for $i\in\{2,\ldots,n\}$, $\lambda'[i]=\delta^{\textup{\texttt{con}}}(\lambda'[i-1],(e_r,e_d))$, for $e_r,e_d$ such that $\lambda[i-1]=\delta^{\textup{\texttt{con}}}(\lambda[i-2],(e_r,e_d))$. We get a contradiction. To see that $Plays(\mathbb{S}^{\textup{\texttt{con}}}(G))$ contains an infinite play. Take an $e\in Act$, then $(E,\pi_1(e)),(E,\pi_1(e)),\ldots$ is a play in $Plays(\mathbb{S}^{\textup{\texttt{con}}}(G))$ generated by the infinite repetition of the action profile $(e,e)$. We conclude that $Plays(\mathbb{S}^{\textup{\texttt{con}}}(G))$ must be infinite and contains an infinite play, and so, by (2.), $Plays(\mathbb{S}^{\textup{\texttt{gen}}}(G))$ must be an infinite set and contain an infinite play.
\end{proof}


\noindent\textbf{Proposition 3.7}

\noindent(1.)~For any $\lambda\in Plays(\mathbb{S}^{\textup{\texttt{tb}}}(G),((E,v),r))$ and an odd $i< length(\lambda)$, $\mathbb{L}(\lambda[i])\neq \mathbb{L}(\lambda[i+1])$. 
\begin{proof}
In turn-based sabotage games, at odd $i$ in $((E,v),r)$-play $\lambda$ it is demon's turn to move, say that it removes $e$. Then $\mathbb{L}({\lambda(i)})\setminus \mathbb{L}({\lambda(i+1)})=q_e$. 
\end{proof}

\noindent(2.)~For any $\lambda \in Plays(\mathbb{S}^{\textup{\texttt{tb}}}(G))$, there is a $\lambda'\in Plays(\mathbb{S}^{\textup{\texttt{gen}}}(G))$, such that for all $i\leq length(\lambda)$, $\mathbb{L}(\lambda[i])= \mathbb{L}(\lambda'[i])$. 
\begin{proof}
Take $\lambda\in Plays(\mathbb{S}^{\textup{\texttt{tb}}}(G))$. W.l.o.g., assume that runner moves first, then, by Prop.~3.5.1., $\lambda$ is a finite sequence $((E^0,v^0),r),\ldots,((E^n,v^n),r)$. We claim that the required play is any $\lambda'$ such that $\lambda'[0,n]=(E^0,v^0),\ldots,(E^n,v^n)$. We need to show that for all $i\leq length(\lambda)$, $\lambda'[i+1]=\delta^{\textup{\texttt{gen}}}(\lambda'[i],\alpha)$, for some $\alpha \in act^{\textup{\texttt{gen}}}(\lambda'[i])$. In fact, for each $a\in \mathbb{AG}$ and for all $t=((E,v),a)\in S^{\textup{\texttt{tb}}}$ and $s=(E,v)\in S^{\textup{\texttt{gen}}}$, $act^{\textup{\texttt{tb}}}(t)\subseteq act^{\textup{\texttt{gen}}}(s)$ (for any state all action profiles allowed in the turn-based game are allowed in the general game), and that for all $\alpha \in act^{\textup{\texttt{tb}}}(((E,v),a))$ and $\delta^{\textup{\texttt{tb}}}(((E,v),a),\alpha))=(((E'v'),a'))$ then $\delta^{\textup{\texttt{gen}}}((E,v),\alpha))=((E',v'))$. Finally, observe that for all $a\in\mathbb{AG}$, $((E,v),a)\in S^{\textup{\texttt{tb}}}$ and $(E,v)\in S^{\textup{\texttt{gen}}}$, $\mathbb{L}(((E,v),a))=\mathbb{L}((E,v))$ (structural labeling $\mathbb{L}$ only takes into account the graph structure, edges and runner's position, of the game-state).
\end{proof}


\noindent\textbf{Proposition 3.13}
Let $G=(V,E)$ be a graph. Runner has a winning strategy in a turn-based RSG $(G,v_0,v_g)$ \ iff \ $\mathbb{M}^{R\textup{\texttt{tb}}},((E,v_0),r)\models \langle\!\langle \{r\} \rangle\!\rangle F g$.

\begin{proof} The following are equivalent:
\begin{enumerate}
\item $\mathbb{M}^{R\textup{\texttt{tb}}},((E,v_0),r)\models \langle\!\langle \{r\} \rangle\!\rangle F g$ 
\item there is an $r$-strategy $str_{r}$, s.t. $\mathbb{M}^{R\textup{\texttt{tb}}},\lambda\models F g$ for all $\lambda\in Plays(\mathbb{S}^{\textup{\texttt{tb}}}(G),v_0,str_{r})$;

\item there is an $r$-strategy $str_{r}$, s.t. $\mathbb{M}^{R\textup{\texttt{tb}}},\lambda\models \top U g$ for all $\lambda\in Plays(\mathbb{S}^{\textup{\texttt{tb}}}(G),v_0,str_{r})$;

\item there is an $r$-strategy $str_{r}$, s.t. for all $\lambda\in Plays(\mathbb{S}^{\textup{\texttt{tb}}}(G),v_0,str_{r})$ there is $i\geq 0$, s.t. $\mathbb{M}^{R\textup{\texttt{tb}}},\lambda[i]\models g$;

\item there is an $r$-strategy $str_{r}$, s.t. for all $\lambda\in Plays(\mathbb{S}^{\textup{\texttt{tb}}}(G),v_0,str_{r})$ there is $i\geq 0$, s.t. $g\in \mathbb{L}(\lambda[i])$;

\item there is an $r$-strategy $str_{r}$, s.t. for all $\lambda\in Plays(\mathbb{S}^{\textup{\texttt{tb}}}(G),v_0,str_{r})$ there is $i\geq 0$, s.t. $\lambda[i]=((E,v),a)$ for some $a\in \mathbb{A}$ with $v=v_g$;

\item there is a function $str_r:S^{\textup{\texttt{tb}}}\rightarrow Act^{\textup{\texttt{skip}}}$ with $str_r(s)\in act^{\textup{\texttt{tb}}}(r,s)$, s.t. for all $\lambda\in Plays(\mathbb{S}^{\textup{\texttt{tb}}}(G),v_0,str_{r})$ there is  $i\geq 0$, s.t. $\lambda[i]=((E,v),a)$ for some $a\in \mathbb{A}$ with $v=v_g$;

\item there is a function $str_r:S^{\textup{\texttt{tb}}}\rightarrow Act^{\textup{\texttt{skip}}}$ with $str_r(s)\in act^{\textup{\texttt{tb}}}(r,s)$, s.t. for all $\lambda\in Plays(\mathbb{S}^{\textup{\texttt{tb}}}(G),v_0)$ such that $\lambda[j+1]\in \{\delta^{\textup{\texttt{tb}}}(\lambda[j], \alpha)\mid \alpha\in act^{\textup{\texttt{tb}}}(\lambda[j]) \text{ and } str_r(\lambda[j])\sqsubseteq \alpha\}$ there is an $i\geq 0$, s.t. $\lambda[i]=((E,v),a)$ for some $a\in \mathbb{A}$ with $v=v_g$;

\item there is a function $str_r:S^{\textup{\texttt{tb}}}\rightarrow Act^{\textup{\texttt{skip}}}$ with $str_r(s)\in act^{\textup{\texttt{tb}}}(r,s)$, s.t. for all $\lambda\in Plays(\mathbb{S}^{\textup{\texttt{tb}}}(G),v_0)$ such that for even $j<length(\lambda)-1$, $\lambda[j+1]=\delta^{\textup{\texttt{tb}}}(\lambda[j],(str_r(\lambda[j]), \texttt{skip}))$ and for odd $j<length(\lambda)-1$, $\lambda[j+1]\in\{\delta^{\textup{\texttt{tb}}}(\lambda[j],(\texttt{skip},e))\mid e\in Edges(\lambda[j])\}$ there is an $i\geq 0$, s.t. $\lambda[i]=((E,v),a)$ for some $a\in \mathbb{A}$ with $v=v_g$;

\item there is a function $win:\mathcal{S}(G)\rightarrow E$ that runner can apply at her choice point $s^k$ such that whichever remaining edge is removed from, runner retains the connectivity to the goal from $\pi_2(win(s))$ in $s^{k+1}$. Namely: $win((E',v))=str_r((E',v),r)$. So, runner has a winning strategy in RSG $(G,v_0,v_g)$.
\end{enumerate}
\end{proof}

\noindent\textbf{Proposition 3.15}
Let $((V,E),v_0,v_g)$ be a RSG. If  $\mathbb{M}^{R\textup{\texttt{con}}},(E,v_0)\models \langle\!\langle \{r,d\} \rangle\!\rangle  F g$, then $\mathbb{M}^{R\textup{\texttt{tb}}},((E,v_0),r)\models \langle\!\langle \{r,d\} \rangle\!\rangle  F g$.

\begin{proof}
Assume that $\mathbb{M}^{R\textup{\texttt{con}}},(E,v_0)\models \langle\!\langle \{r,d\} \rangle\!\rangle  F g$, which means that  there is pair of strategies $(str_r, str_d)$, with $str_a:S^{\textup{\texttt{con}}}\rightarrow Act$, and $str_a\in act^{\textup{\texttt{con}}}(s)$ (for $a\in\mathbb{AG}$), s.t. for all $\lambda\in Plays(\mathbb{S}^{\textup{\texttt{con}}}(G),v_0)$ such that $\lambda[j+1]=\delta^{\textup{\texttt{con}}}(\lambda[j], (str_r(\lambda[j]),str_d(\lambda[j])))$, there is an $i\geq 0$, s.t. $\lambda[i]=(E,v)$ with $v=v_g$. Let us define $str'_r, str'_d$ with $str'_a:S^{\textup{\texttt{tb}}}\rightarrow Act^{\texttt{skip}}$, and $str'_a\in act^{\textup{\texttt{tb}}}(s)$ (for $a\in\mathbb{AG}$) in the following way:
\begin{itemize}
\item $str'_r(((E,v),r)):=str_r((E,v))$ and $str'_r(((E,v),d)):=\texttt{skip}$; 
\item $str'_d(((E,v),d)):=str_d((E,v))$ and $str'_d(((E,v),r)):=\texttt{skip}$. 
\end{itemize}
First, note that $str'_a:S^{\textup{\texttt{tb}}}\rightarrow Act^{\texttt{skip}}$, and $str'_a\in act^{\textup{\texttt{tb}}}(s)$ (for $a\in\mathbb{AG}$). We also have that for all $\lambda\in Plays(\mathbb{S}^{\textup{\texttt{tb}}}(G),v_0)$, such that $\lambda[j+1]=\delta^{\textup{\texttt{tb}}}(\lambda[j], (str'_r(\lambda[j]),str'_d(\lambda[j])))$, there is an $i\geq 0$, s.t. $\lambda[i]=((E,v),a)$ for some $a\in \mathbb{A}$ with $v=v_g$.
\end{proof}

\noindent\textbf{Proposition 3.16}
      Let $G=(V,E)$ be a graph. Runner has a winning strategy in a turn-based LSG $(G,v_0,b)$ \ iff \ $\mathbb{M}^{R\textup{\texttt{tb}}},((E,v_0),r)\models \langle\!\langle \{r\} \rangle\!\rangle \top  U_{2b}  (X\top)$.  
\begin{proof}
Analogous to the proof of Prop.~3.13.
\end{proof}



\subsection*{Examples of rooted sabotage game structures}

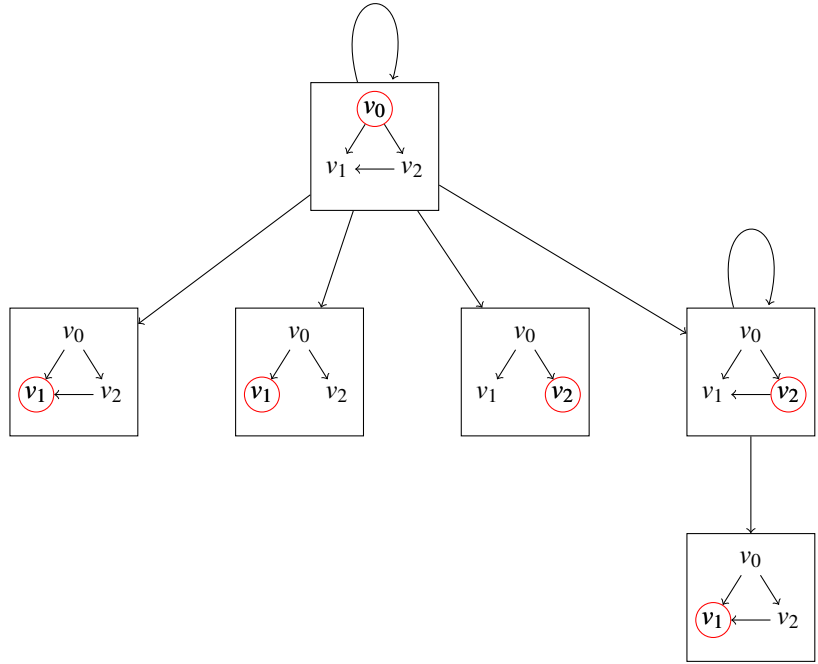
\begin{figure}[H]
    \centering
    \begin{tikzpicture}[
node distance=1.5cm,
every node/.style={font=\small},
g/.style={draw, rectangle, minimum width=1.7cm, minimum height=1.7cm},
v/.style={circle, inner sep=1pt},
redv/.style={circle, draw=red, inner sep=1pt}
]

\newcommand{\gadget}[5]{
\node[g] (#1) at #2 {};

\node[v] (#1v0) at ($(#1.center)+(0,0.5)$) {$v_0$};
\node[v] (#1v1) at ($(#1.center)+(-0.5,-0.3)$) {$v_1$};
\node[v] (#1v2) at ($(#1.center)+(0.5,-0.3)$) {$v_2$};

\ifnum#3=0
\node[redv] at (#1v0) {$v_0$};
\fi
\ifnum#3=1
\node[redv] at (#1v1) {$v_1$};
\fi
\ifnum#3=2
\node[redv] at (#1v2) {$v_2$};
\fi

\draw[->] (#1v0) -- (#1v1);
\draw[->] (#1v0) -- (#1v2);

\ifnum#5=1
\draw[->] (#1v2) -- (#1v1);
\fi
}

\gadget{root}{(0,0)}{0}{1}{1}

\gadget{g1}{(-4,-3)}{1}{0}{1}
\gadget{g2}{(-1,-3)}{1}{0}{0}
\gadget{g3}{(2,-3)}{2}{0}{0}
\gadget{g4}{(5,-3)}{2}{1}{1}

\gadget{g5}{(5,-6)}{1}{0}{1}

\draw[->] (root) -- (g1);
\draw[->] (root) -- (g2);
\draw[->] (root) -- (g3);
\draw[->] (root) -- (g4);
\draw[->] (g4) -- (g5);
\draw[->] (root) edge[loop above] (root);
\draw[->] (g4) edge[loop above] (g4);
\end{tikzpicture}
    \caption{Concurrent sabotage game structure rooted in $(E,v_0)$. Red circle marks runner's position in a game-state. Each transition corresponds to a concurrent action in a game-state. Note that whenever both agents choose the same edge in a joint concurrent move, the action is canceled, hence the reflexive arrows. It is easy to see there are finite and infinite plays}
    \label{fig:concurrentGraph}
\end{figure}

\begin{figure}[h]
    \centering
    \begin{tikzpicture}[
every node/.style={font=\small},
g/.style={draw, rectangle, minimum width=1.7cm, minimum height=1.7cm},
v/.style={circle, inner sep=1pt},
redv/.style={circle, draw=red, inner sep=1pt}
]

\newcommand{\gadget}[6]{
\node[g] (#1) at #2 {};

\node[v] (#1v0) at ($(#1.center)+(0,0.6)$) {$v_0$};
\node[v] (#1v1) at ($(#1.center)+(-0.6,-0.3)$) {$v_1$};
\node[v] (#1v2) at ($(#1.center)+(0.6,-0.3)$) {$v_2$};

\ifnum#3=0 \node[redv] at (#1v0) {$v_0$}; \fi
\ifnum#3=1 \node[redv] at (#1v1) {$v_1$}; \fi
\ifnum#3=2 \node[redv] at (#1v2) {$v_2$}; \fi

\ifnum#4=1 \draw[->] (#1v0) -- (#1v1); \fi
\ifnum#5=1 \draw[->] (#1v0) -- (#1v2); \fi
\ifnum#6=1 \draw[->] (#1v2) -- (#1v1); \fi
}

\gadget{r}{(0,0)}{0}{1}{1}{1}

\gadget{a}{(-2,-3)}{1}{1}{1}{1}
\gadget{b}{(2,-3)}{2}{1}{1}{1}

\draw[->] (r) -- (a);
  \draw[->]        (r) -- (b);

\gadget{a1}{(-5,-6)}{1}{1}{0}{1}
\gadget{a2}{(-3,-6)}{1}{1}{0}{1}
\gadget{a3}{(-1,-6)}{1}{1}{1}{0}

\draw[->] (a) -- (a1);
\draw[->]          (a) -- (a2);
 \draw[->]         (a) -- (a3);

\gadget{b1}{(1,-6)}{2}{0}{0}{1}
\gadget{b2}{(3,-6)}{2}{1}{1}{0}
\gadget{b3}{(5,-6)}{2}{0}{0}{1}

\draw[->] (b) -- (b1);
 \draw[->]          (b) -- (b2);
  \draw[->]         (b) -- (b3);

\gadget{c1}{(1,-9)}{1}{1}{0}{1}
\gadget{c2}{(5,-9)}{1}{1}{0}{1}

\draw[->] (b1) -- (c1);
 \draw[->]          (b3) -- (c2);

\gadget{d1}{(1,-12)}{1}{1}{0}{0}
\gadget{d2}{(3,-12)}{1}{1}{0}{1}
\gadget{d3}{(5,-12)}{1}{1}{0}{0}

\draw[->] (c1) -- (d1);
   \draw[->]        (c1) -- (d2);
   \draw[->]        (c2) -- (d2);
    \draw[->]       (c2) -- (d3);

\end{tikzpicture}
    \caption{Turn-based sabotage game structure rooted in $(E,v_0)$. Red circle marks the runner's position in a state. The control over the levels of the structure alternates between runner and demon, starting with runner at the top}
    \label{fig:turnBasedGraph}
\end{figure}
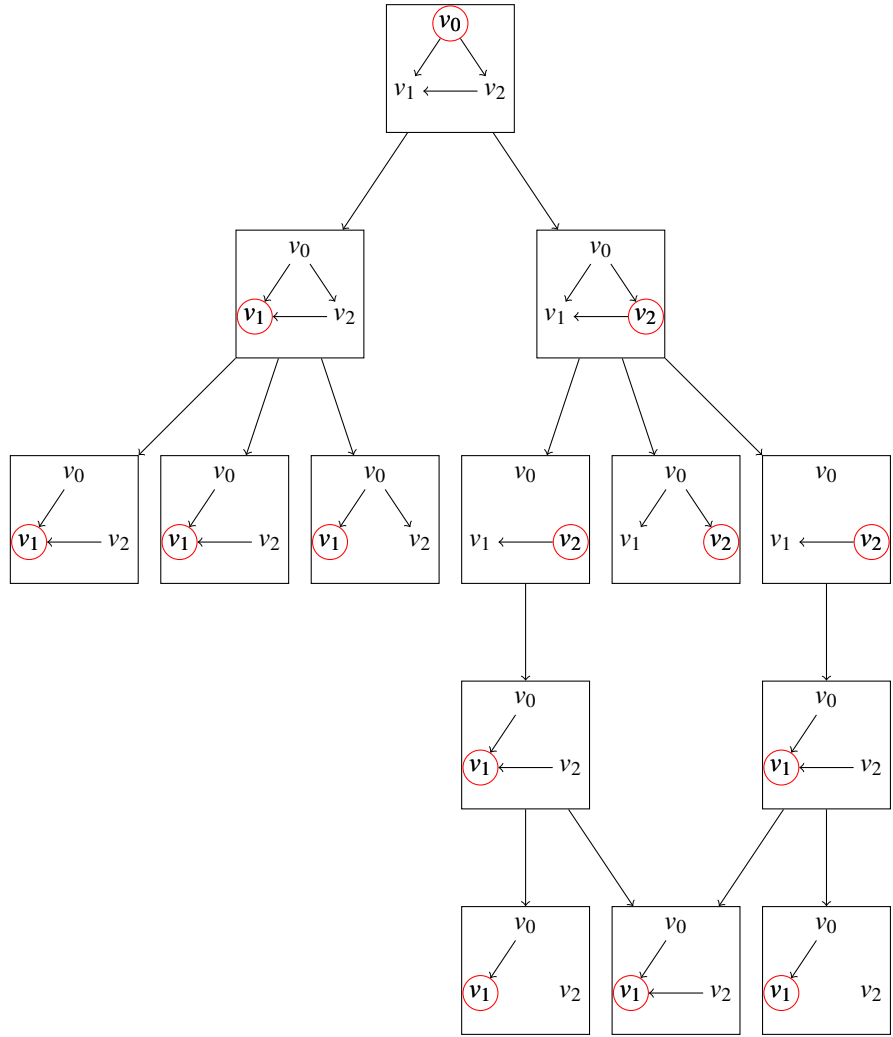

\end{document}